\documentclass[a4paper,oneside,12pt]{article}
\usepackage{amsmath,amsfonts,amscd,amssymb}
\usepackage{longtable,geometry}
\usepackage[english]{babel}
\usepackage[active]{srcltx}
\usepackage{graphicx}
\usepackage{pstricks}
\usepackage{bbm}
\usepackage{stmaryrd}
\newtheorem{theorem}{Theorem}

\newtheorem{lemma}{Lemma}

\newtheorem{definition}{Definition}
\newtheorem{remark}{Remark}
\newtheorem{conjecture}{Conjecture}

\newcommand{\g}{\gamma}
\newcommand{\He}{\mathbb{H}}
\newcommand{\e}{{\rm e}}
\newcommand{\I}{{\rm i}}
\newcommand{\C}{\mathbb{C}}
\newcommand{\ep}{\varepsilon}
\newcommand{\W}{{\rm W}}
\newenvironment{proof}[1][\relax]%
  {\paragraph{Proof\ifx#1\relax\else~of #1\fi}}%
  {~\hfill$\square$\par\bigskip}

\title{The connective constant of the honeycomb lattice equals $\sqrt{2+\sqrt 2}$}
\author{Hugo Duminil-Copin and Stanislav Smirnov}
\date{}

\begin{document}

\maketitle

\begin{abstract}We provide the first mathematical proof that the connective constant of the hexagonal lattice is equal to $\sqrt{2+\sqrt{2}}$. This value has been derived non rigorously by B. Nienhuis in 1982, using Coulomb gas approach from theoretical physics. Our proof uses a parafermionic observable for the self avoiding walk,
which satisfies a half of the discrete Cauchy-Riemann relations.
Establishing the other half of the relations (which conjecturally holds in the scaling limit)
would also imply convergence of the self-avoiding walk to SLE($8/3$).
 \end{abstract}

\section{Introduction}
A famous chemist P. Flory \cite{Flory} proposed to consider self-avoiding (\emph{i.e.} visiting every vertex at most once) walks on a lattice as a model for polymer chains. Self-avoiding walks turned out to be a very interesting object, leading to rich mathematical theories and challenging questions, see \cite{MadrasSlade}.

Denote by $c_n$ the number of $n$-step self-avoiding walks on the hexagonal lattice $\He$ started from some fixed vertex, \emph{e.g.} the origin. Elementary bounds on $c_n$ (for instance $\sqrt{2}^n\leq c_n\leq 3\cdot 2^{n-1}$) guarantee that $c_n$ grows exponentially fast. Since a $(n+m)$-step self-avoiding walk can be uniquely cut into a $n$-step self-avoiding walk and a parallel translation of a $m$-step self-avoiding  walk, we infer that
\begin{equation*}
	c_{n+m}\leq c_nc_m,
\end{equation*}
from which it follows that there exists $\mu\in(0,+\infty)$ such that
\begin{equation*}	
	\mu:=\lim_{n\rightarrow \infty}c_n^{\ \frac 1n}.
\end{equation*}
The positive real number $\mu$ is called the \emph{connective constant} of the hexagonal lattice. 

Using Coulomb gas formalism, B. Nienhuis \cite{Nienhuis,Nienhuis-jsp} proposed physical arguments for $\mu$ to have the value $\sqrt{2+\sqrt 2}$. We rigorously prove this statement. While our methods are different from those applied by Nienhuis, they are similarly motivated by considerations of vertex operators in the $O(n)$ model.
Our methods do not directly apply to the square lattice,
for which the value of the connective constant is different and currently unknown.

\begin{theorem}\label{theorem}
For the hexagonal lattice, $$\mu=\sqrt{2+\sqrt{2}}.$$
\end{theorem}

It will be convenient to consider walks between \emph{mid-edges} of $\He$, \emph{i.e.} centers of edges of $\He$  (the set of mid-edges will be denoted by $H$). We will write $\g:a\rightarrow E$ if a walk $\g$ starts at $a$ and ends at some mid-edge of $E\subset H$. In the case $E=\{b\}$, we simply write $\g:a\rightarrow b$. The \emph{length} $\ell(\g)$ of the walk is the number of vertices visited by $\g$.

We will work with the partition function
\begin{equation*}
	Z(x)=\sum_{\g\ :\ a\rightarrow H}x^{\ell(\g)}\quad\in(0,+\infty].
\end{equation*}
This sum does not depend on the choice of $a$, and is increasing in $x$. 
Establishing the identity $\mu=\sqrt{2+\sqrt 2}$ is equivalent to showing that 
$Z(x)=+\infty$ for $x>1/\sqrt{2+\sqrt{2}}$ and $Z(x)<+\infty$ for $x<1/\sqrt{2+\sqrt 2}$. 
To this end, we analyze walks restricted to bounded domains and weighted
depending on their winding. The modified sum can be defined as 
a \emph{parafermionic observable} arising from a disorder operator. 
Such observables exist for other models, 
see \cite{CardyIkhlef,chelkak-smirnov-iso,smirnov-icm2010}. 

The paper is organized as follows. In Section 2, the parafermionic observable is introduced and its key property is derived. Section 3 contains the proof of Theorem~\ref{theorem}. Section 4 discusses conformal invariance conjectures for self-avoiding walks.
To simplify formul\ae, below we set  $x_c:=1/\sqrt{2+\sqrt 2}$ and $j=\e^{\I2\pi/3}$.

\section{Parafermionic observable}

A (hexagonal lattice) \emph{domain} $\Omega\subset H$ is a union of all mid-edges emanating from a given collection of vertices $V(\Omega)$ (see Fig. \ref{fig:SAWpicture}): a mid-edge $z$ belongs to $\Omega$ if at least one end-point of its associated edge is in $\Omega$, it belongs to $\partial \Omega$ if only one of them is in $\Omega$. We further assume $\Omega$ to be simply connected, \emph{i.e.} having a connected complement. 

\begin{figure}[ht]
    \begin{center}
      \includegraphics[width=0.6\hsize]{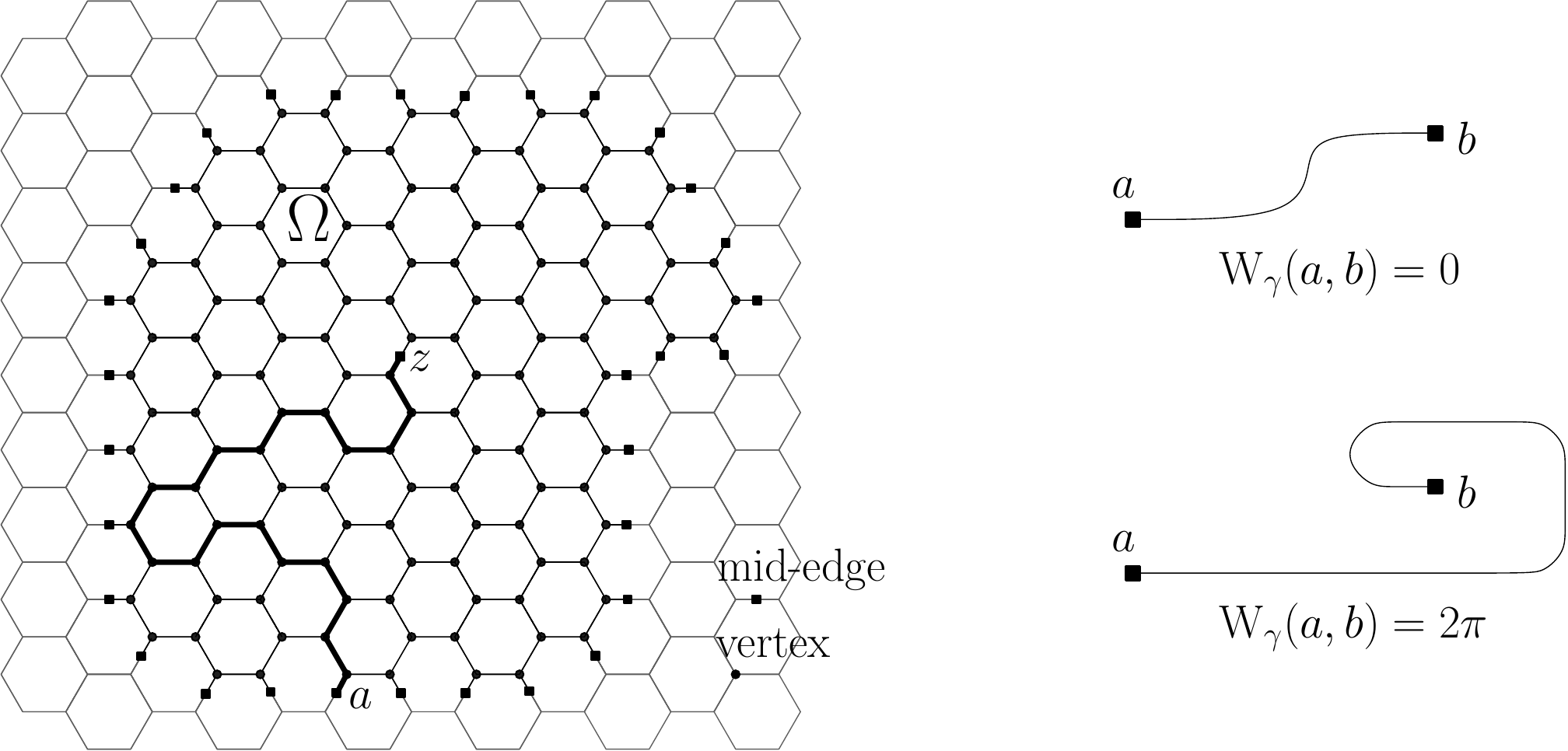}
     \end{center}
    \caption{\textbf{Left.} A domain $\Omega$ with boundary mid-edges 
labeled by small black squares, and vertices of $V(\Omega)$ labeled by circles. 
\textbf{Right.} Winding of a curve $\g$.}
    \label{fig:SAWpicture} 
  \end{figure}
  
For a self-avoiding walk $\g$ between mid-edges $a$ and $b$ (not necessarily the start and the end), we define its
\emph{winding $\W_\g(a,b)$} as the total rotation of the direction in radians when $\g$ is traversed from $a$ to $b$, see Fig.~\ref{fig:SAWpicture}. 

Our main tool is given by the following
\begin{definition}
The \emph{parafermionic observable} for $a\in \partial \Omega$, $z\in \Omega$, 
is defined by
\begin{equation*}
	F(z)=F(a,z,x,\sigma)=\sum_{\g\subset \Omega:\ a\rightarrow z}\e^{-\I\sigma \W_\g(a,z)} x^{\ell(\g)}.
\end{equation*}
\end{definition}

\begin{lemma}\label{lem:relation}
	If $x=x_c$ and $\sigma=\frac 58$, then $F$ satisfies the following relation for every vertex $v\in V(\Omega)$:
	\begin{equation}\label{relation around vertex}
		(p-v)F(p)+(q-v)F(q)+(r-v)F(r)=0,
	\end{equation}
	where $p,q,r$ are the mid-edges of the three edges adjacent to $v$.
\end{lemma}

Note that with $\sigma=5/8$, the complex weight $\e^{-\I\sigma \W_\g(a,z)}$ 
can be interpreted as a product of terms $\lambda$ or $\bar{\lambda}$ 
per left or right turn of $\g$ drawn from $a$ to $z$, with
 $$\lambda=\exp \left(-\I \frac{5}{8}\cdot \frac{\pi}{3}\right)=
\exp \left(-\I \frac{5\pi}{24}\right).$$

\begin{proof}
We start by choosing notation so that $p$, $q$ and $r$ 
follow counter-clockwise around $v$. 
Note that the left-hand side of \eqref{relation around vertex}
can be expanded into the sum of contributions $c(\g)$ 
of all possible walks $\g$ finishing at $p,q$ or $r$. 
For instance, if a walk ends at the mid-edge $p$, its contribution 
will be given by
$$c(\g)=(p-v)\cdot \e^{-\I\sigma \W_{\g}(a,p)}~x_c^{\ell(\g)}.$$
One can partition the set of walks $\g$ finishing at $p,q$ or $r$ into pairs and triplets of walks in the following way, see Fig \ref{fig:pairs}:
\begin{itemize}
\item If a walk $\g_1$ visits all three mid-edges $p,q,r$, it means that the edges belonging to $\g_1$ form a disjoint self-avoiding path plus (up to a half-edge) a self-avoiding loop from $v$ to $v$. One can associate to $\g_1$ the walk passing through the same edges, but exploring the loop from $v$ to $v$ in the other direction. Hence, walks visiting the three mid-edges can be grouped in pairs.

\item If a walk $\g_1$ visits only one mid-edge, it can be associated to two walks $\g_2$ and $\g_3$ that visit exactly two mid-edges by prolonging the walk one step further (there are two possible choices). The reverse is true: a walk visiting exactly two mid-edges is naturally associated to a walk visiting only one mid-edge by erasing the last step. Hence, walks visiting one or two mid-edges can be grouped in triplets.
\end{itemize}
If one can prove that the sum of contributions to \eqref{relation around vertex}
of each pair or triplet vanishes, then their total sum is zero, and \eqref{relation around vertex} holds. 

Let $\g_1$ and $\g_2$ be two associated walks as in the first case. 
Without loss of generality, we may assume that $\g_1$ ends at $q$ and $\g_2$ ends at $r$. Note that $\g_1$ and $\g_2$ coincide up to the mid-edge $p$,
and then follow an almost complete loop in two opposite directions.
It follows that
$$\ell(\g_1)=\ell(\g_2)\quad\quad \text{and}\quad \quad \left\{\ \substack{
	\W_{\g_1}(a,q)=\W_{\g_1}(a,p)+\W_{\g_1}(p,q)=\W_{\g_1}(a,p)-\frac{4\pi}{3}\\ \ \\
	\W_{\g_2}(a,r)=\W_{\g_2}(a,p)+\W_{\g_2}(p,r)=\W_{\g_1}(a,p)+\frac{4\pi}{3}}\right.\quad.$$
In order to evaluate the winding of $\g_1$ between $p$ and $q$ above, 
we used the fact that $a$ is on the boundary and $\Omega$ is simply connected. 
We conclude that
\begin{align*}c(\g_1)+c(\g_2)&=(q-v)\e^{-\I\sigma \W_{\g_1}(a,q)} x_c^{\ell(\g_1)}+(r-v)\e^{-\I\sigma \W_{\g_2}(a,r)} x_c^{\ell(\g_2)}\\
&=(p-v)\e^{-\I\sigma \W_{\g_1}(a,p)}x_c^{\ell(\g_1)}\left(j\bar{\lambda}^4+\bar{j}\lambda^4\right)=0
\end{align*}
where the last equality holds since $j\bar{\lambda}^4=-i$  by our choice
of $\lambda=\exp (-\I5\pi/24)$. 

Let $\g_1,\g_2,\g_3$ be three walks matched as in the second case. Without loss of generality, we assume that $\g_1$ ends at $p$ and that $\g_2$ and $\g_3$ extend $\g_1$ to $q$ and $r$ respectively. As before, we easily find that

$$\ell(\g_2)=\ell(\g_3)=\ell(\g_1)+1\quad\quad\text{and}\quad\quad\left\{\ \substack{\W_{\g_2}(a,r)=\W_{\g_2}(a,p)+\W_{\g_2}(p,q)=\W_{\g_1}(a,p)-\frac{\pi}{3}\\ \ \\
	\W_{\g_3}(a,r)=\W_{\g_3}(a,p)+\W_{\g_3}(p,r)=\W_{\g_1}(a,p)+\frac{\pi}{3}}\right.\quad.$$
Plugging these values into the respective contributions, we obtain
\begin{align*}c(\g_1)+c(\g_2)+c(\g_3)&=(p-v)\e^{-\I\sigma \W_{\g_1}(a,p)}x_c^{\ell(\g_1)}\left(1+x_cj\bar{\lambda}+x_c\bar{j}\lambda\right)=0.
\end{align*}
Above is the \emph{only} place where we use that $x$ takes its critical value, i.e. 
$x_c^{-1}=\sqrt{2+\sqrt2}=(2\cos \frac{\pi}{8})$. 

The claim of the lemma follows readily by summing over all pairs and triplets.
\end{proof}

\begin{figure}[ht]
    \begin{center}
      \includegraphics[width=0.90\hsize]{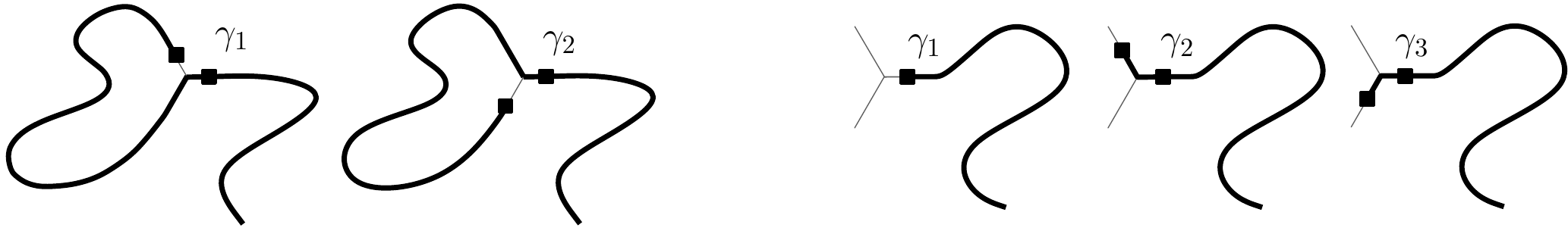}
     \end{center}
    \caption{\textbf{Left: } a pair of walks visiting all the three mid-edges emanating from $v$ and differing by rearranged connections at $v$. 
\textbf{Right:} a triplet of walks, one visiting one mid-edge, the two others visiting two mid-edges, and obtained by prolonging the first one through $v$.}
    \label{fig:pairs} 
  \end{figure} 
\begin{remark}\label{rem:CR}
	Coefficients in \eqref{relation around vertex} are three cube roots of unity multiplied by $p-v$, so its left-hand side can be seen as a discrete $dz$-integral along an elementary contour on the dual lattice. The fact that the integral of the parafermionic observable along discrete contours vanishes suggests that it is 
discrete holomorphic and that self-avoiding walks have a conformally invariant scaling limit, see Section 4. 
\end{remark}

\section{Proof of Theorem~\ref{theorem}}
\paragraph{Counting argument in a strip domain.} We consider a vertical strip domain $S_T$ composed of $T$ strips of hexagons, and its finite version $S_{T,L}$ cut at heights $\pm L$ at angles  $\pm \pi/3$, see Fig. \ref{fig:domain}. Namely, position a hexagonal lattice $\He$ of meshsize 1 in $\C$ so that there exists a horizontal edge $e$ with mid-edge $a$ being 0. Then
\begin{align*}
	V(S_T)&=\{z\in V(\He):0\leq {\rm Re}(z)\leq \frac{3T+1}{2} \},\\
	V(S_{T,L})&=\{z\in V(S_T):|\sqrt 3~{\rm Im}(z)-{\rm Re}(z)|\leq 3L\}.
\end{align*}
Denote by $\alpha$ the left boundary of $S_{T}$, by $\beta$ the right one. Symbols $\ep$ and $\bar{\ep}$ denote the top and bottom boundaries of $S_{T,L}$. Introduce the following (positive) partition functions:
$$
	A_{T,L}^x:=\sum_{\substack{\g\subset S_{T,L} :\  a\rightarrow\alpha\setminus \{a\}}}x^{\ell(\gamma)},~~
	B_{T,L}^x:=\sum_{\substack{\g\subset S_{T,L}:\ a\rightarrow\beta}}x^{\ell(\gamma)},~~
	E_{T,L}^x:=\sum_{\substack{\g\subset S_{T,L}:\ a\rightarrow\ep\cup\bar{\ep}}}x^{\ell(\gamma)}.
$$
In the next lemma, we deduce from relation \eqref{relation around vertex}
a global identity without the complex weights.
\begin{figure}[ht]
    \begin{center}
      \includegraphics[width=0.30\hsize]{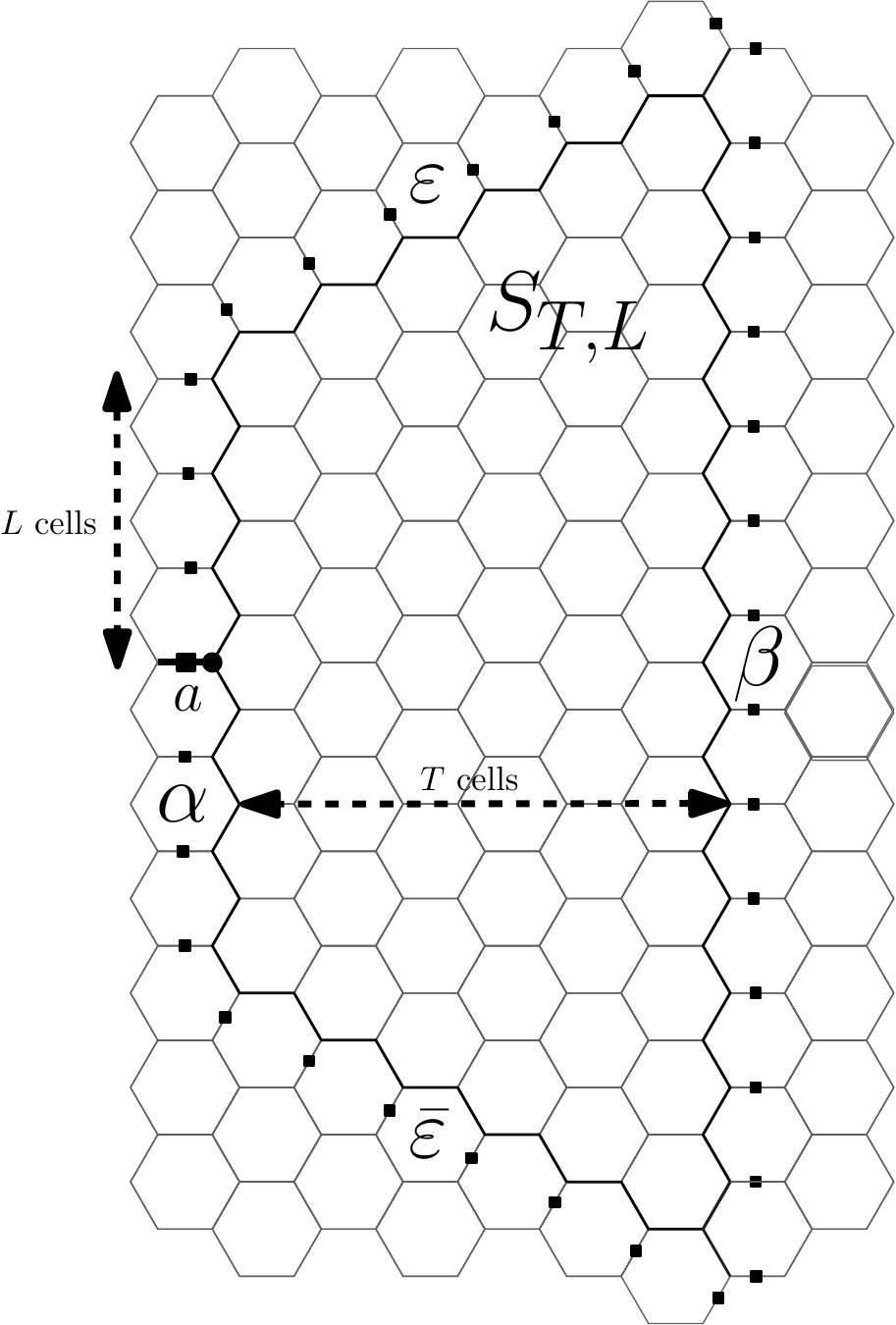}
     \end{center}
    \caption{Domain $S_{T,L}$ and boundary intervals $\alpha$, $\beta$, $\ep$ and $\bar{\ep}$.}
    \label{fig:domain} 
  \end{figure}
\begin{lemma}
	For critical $x=x_c$, the following identity holds
	\begin{equation}\label{equation box}
		1=c_{\alpha}A^{x_c}_{T,L}+B^{x_c}_{T,L}+c_{\ep}E^{x_c}_{T,L},
	\end{equation}
	with positive coefficients $c_{\alpha}=\cos \left(\frac {3\pi}{8}\right)$ and $c_{\ep}=\cos \left(\frac{\pi}{4}\right)$.
\end{lemma}

\begin{proof}
Sum the relation \eqref{relation around vertex} over all vertices in $V(S_{T,L})$. Values at interior mid-edges disappear and we arrive at the identity
\begin{equation}\label{sum}
	0=-\sum_{z\in \alpha}F(z)+\sum_{z\in \beta}F(z)+j\sum_{z\in \ep}F(z)+\bar{j}\sum_{z\in \bar{\ep}}F(z).
\end{equation}
The symmetry of our domain implies that $F(\bar{z})=\bar{F}(z)$, 
where $\bar{x}$ denotes the complex conjugate of $x$.
Observe that the winding of any self-avoiding walk from $a$ to the bottom part of $\alpha$ is $-\pi$ while the winding to the top part is $\pi$. Thus
\begin{align*}
	\sum_{z\in \alpha}F(z)&=F(a)+\sum_{z\in \alpha\setminus \{a\}}F(z)
=F(a)+\frac12\sum_{z\in \alpha\setminus \{a\}}\left(F(z)+F(\bar z)\right)\\
&=1+\frac{\e^{-\I\sigma \pi}+\e^{\I\sigma \pi}}{2}A^x_{T,L}=1-\cos \left(\frac{3\pi}{8}\right)\ A^x_{T,L}=1-c_{\alpha}A^x_{T,L}.
\end{align*}
Above we have used the fact that the only walk from $a$ to $a$ is a trivial one
of length $0$, and so $F(a)=1$. Similarly, the winding from $a$ to any half-edge in $\beta$ (resp. $\ep$ and $\bar{\ep}$) is 0 (resp. $\frac {2\pi}3$ and $-\frac{2\pi}3$), therefore 
$$\sum_{z\in \beta}F(z)=B^x_{T,L}\quad \text{and}\quad j\sum_{z\in \ep}F(z)+\bar{j}\sum_{z\in \bar{\ep}}F(z)=\cos  \left(\frac\pi 4\right)\  E^x_{T,L}=c_{\ep}E^x_{T,L}.$$
The lemma follows readily by plugging the last three formul\ae\ into \eqref{sum}.
\end{proof}

Observe that sequences $(A^x_{T,L})_{L>0}$ and $(B^x_{T,L})_{L>0}$ are increasing in $L$ and are bounded for $x\leq x_c$ thanks 
to \eqref{equation box} and their monotonicity in $x$. Thus they have limits
$$
	A^x_T:=\lim_{L\rightarrow \infty}A^x_{T,L}=\sum_{\g\subset S_T:\ a\rightarrow\alpha\setminus \{a\}} x^{\ell(\g)},~~~~
	B^x_T:=\lim_{L\rightarrow \infty}B^x_{T,L}=\sum_{\g\subset S_T:\ a\rightarrow \beta} x^{\ell(\g)}.
$$
Identity \eqref{equation box} then implies that $(E^{x_c}_{T,L})_{L>0}$ decreases and converges to a limit $E^{x_c}_T=\lim_{L\rightarrow \infty}E^{x_c}_{T,L}$. 
Passing to a limit in \eqref{equation box}, we arrive at
\begin{equation}\label{equation strip}
	1=c_{\alpha}A^{x_c}_T+B^{x_c}_T+c_{\ep}E^{x_c}_T.
\end{equation}
  
\begin{proof}[Theorem \ref{theorem}]
We start by proving that $Z(x_c)=+\infty$, and hence $\mu\geq \sqrt{2+\sqrt 2}$. 
Suppose that for some $T$, $E^{x_c}_T>0$. As noted before, 
$E^{x_c}_{T,L}$ decreases in $L$ and so
\begin{equation*}
	Z(x_c)\geq \sum_{L>0}E^{x_c}_{T,L}\geq \sum_{L>0} E^{x_c}_T=+\infty,
\end{equation*}
which completes the proof. 

Assuming on the contrary that $E^{x_c}_T=0$ for all $T$, we simplify  \eqref{equation strip} to
\begin{equation}\label{infinite strip}
	1=c_{\alpha}A_T^{x_c}+B_T^{x_c}.
\end{equation}
Observe that a walk $\gamma$ entering into the count of $A_{T+1}^{x_c}$ 
and not into $A_T^{x_c}$ has to visit some vertex adjacent 
to the right edge of $S_{T+1}$. 
Cutting $\gamma$ at the first such point (and adding half-edges to the two halves), 
we uniquely decompose it into two walks crossing $S_{T+1}$ 
(these walks are usually called bridges), 
which together are one step longer than $\gamma$. 
We conclude that
\begin{equation}\label{rec relation}
	A_{T+1}^{x_c}-A_T^{x_c}\leq x_c\left(B_{T+1}^{x_c}\right)^2.
\end{equation}
Combining \eqref{infinite strip} for two consecutive values of $T$ with 
\eqref{rec relation}, we can write
\begin{align*}
0&=1-1=(c_{\alpha}A_{T+1}^{x_c}+B_{T+1}^{x_c})-(c_{\alpha}A_{T}^{x_c}+B_{T}^{x_c})\\
&=c_{\alpha}(A_{T+1}^{x_c}-A_T^{x_c})+B_{T+1}^{x_c}-B_T^{x_c}
\leq c_{\alpha}x_c\left(B_{T+1}^{x_c}\right)^2+B_{T+1}^{x_c}-B_T^{x_c},
\end{align*}
and so
$$c_{\alpha}x_c\left(B_{T+1}^{x_c}\right)^2+B_{T+1}^{x_c}\geq B_T^{x_c}.$$
It follows easily by induction, that
$$B_T^{x_c}\geq  {\min [B_1^{x_c},1/(c_{\alpha}x_c)]}~/~{T}$$
for every $T\geq1$, and therefore
$$Z(x_c)\geq \sum_{T>0}B_T^{x_c}=+\infty.$$
This completes the proof of the estimate $\mu\geq x_c^{-1}= \sqrt{2+\sqrt2}$.

It remains to prove the opposite inequality $\mu\leq x_c^{-1}$. 
To estimate the partition function from above, we will decompose self-avoiding walks into bridges.
A \emph{bridge} of width $T$ is a self-avoiding walk in $S_T$ from one side to the opposite side, defined up to vertical translation. The partition function of bridges of width $T$ is $B_T^x$, which is at most $1$ by \eqref{equation strip}.
Noting that a bridge of width $T$ has length at least $T$, we obtain for $x<x_c$
\begin{equation*}
	B_T^x\leq \left(\frac{x}{x_c}\right)^TB_T^{x_c}\leq \left(\frac{x}{x_c}\right)^T.
\end{equation*}
Thus, for $x<x_c$, the series $\sum_{T>0} B_T^x$ converges and so does the product $\prod_{T>0} (1+B_T^x)$. Let us assume for the moment the following fact: 
\emph{any self-avoiding walk can be canonically decomposed into a sequence of bridges of widths $T_{-i}<\cdots<T_{-1}$ and $T_0>\cdots>T_j$, and, if one fixes the starting mid-edge and the first vertex visited, the decomposition uniquely determines the walk}. Such decomposition was first introduced by Hammersley and Welsh in \cite{HammersleyWelsh} (for a modern treatment, see Section 3.1 of \cite{MadrasSlade}). Applying the decomposition to walks starting at $a$ (the first visited vertex is 0 or -1), we can estimate
$$Z(x)\leq 2\sum_{\substack{T_{-i}<\cdots<T_{-1}\\
T_j<\cdots<T_0}} \left(\prod_{k=-i}^jB_{T_k}^x\right)=2\prod_{T>0}(1+B_T^x)^2<\infty.
$$
The factor 2 is due to the fact that there are two possibilities for the first vertex once we fix the starting mid-edge. Therefore, $Z(x)<+\infty$ whenever $x<x_c$ and $\mu\leq x_c^{-1}=\sqrt{2+\sqrt2}$. To complete the proof of the theorem it only remains to prove that such a decomposition into bridges does exist. Once again, this fact is well-known \cite{MadrasSlade,HammersleyWelsh}, but we include the proof for completeness.
\begin{figure}[ht]
    \begin{center}
      \includegraphics[width=0.60\hsize]{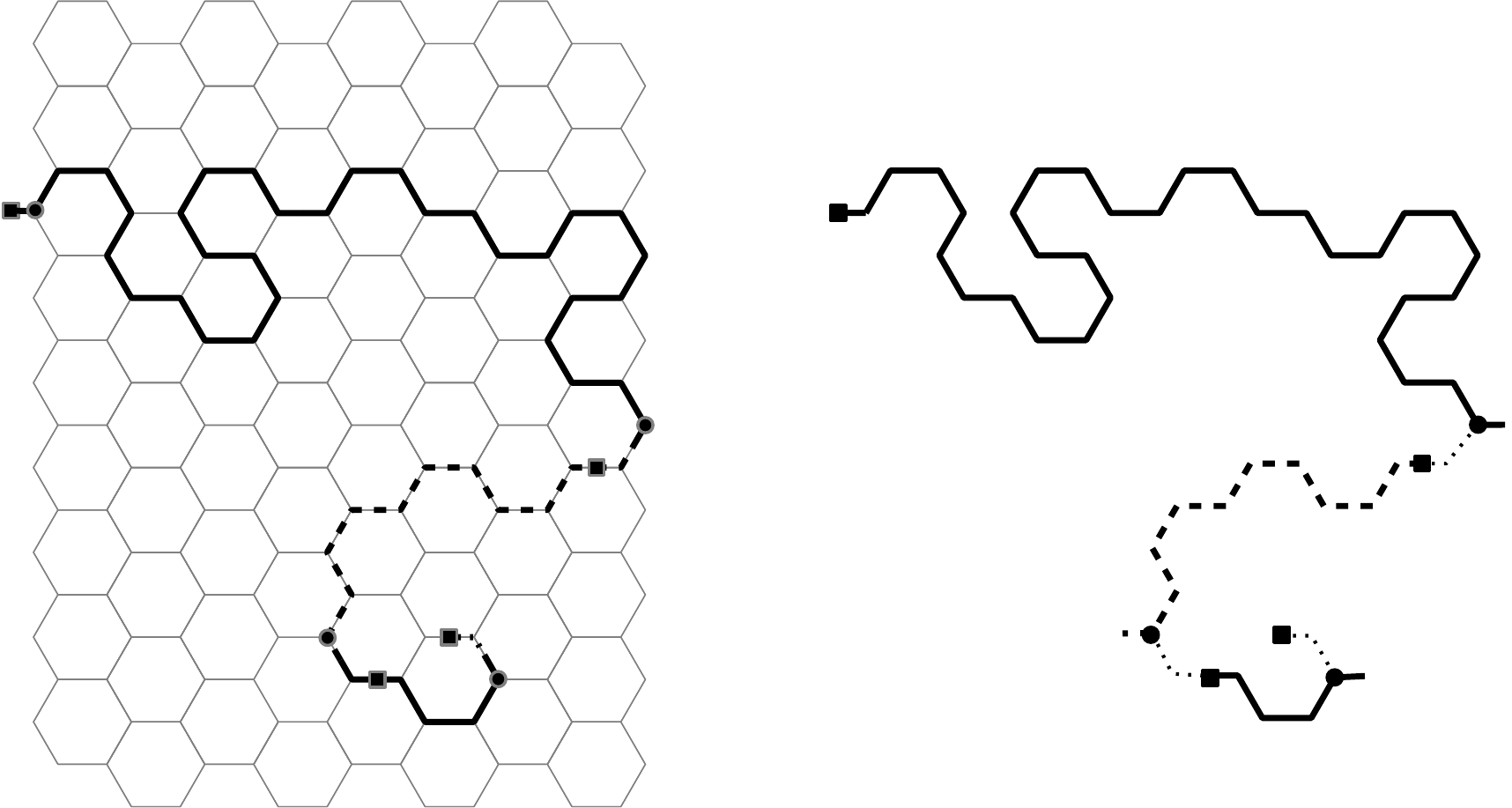}
     \end{center}
    \caption{\textbf{Left:} Decomposition of a half-plane walk into four bridges with widths $8>3>1>0$. The first bridge corresponds to the maximal bridge containing the origin. Note that the decomposition contains one bridge of width 0. \textbf{Right:} The reverse procedure. If the starting mid-edge and the first vertex are fixed, the decomposition is unambiguous.}
    \label{fig:decomposition} 
  \end{figure}

First assume that $\tilde{\g}$ is a half-plane self-avoiding walk, meaning that the start of $\tilde{\g}$ has extremal real part: we prove by induction on the width $T_0$ that the walk admits a canonical decomposition into bridges of widths $T_0>\cdots>T_j$. Without loss of generality, we assume that the start has minimal real part. Out of the vertices having the maximal real part, choose the one visited last, say after $n$ steps. The $n$ first vertices of the walk form a bridge $\tilde{\g}_1$ of width $T_0$, which is the first bridge of our decomposition when prolonged to the mid-edge on the right of the last vertex. We forget about the $(n+1)$-th vertex, since there is no ambiguity in its position. The consequent steps form a half-plane walk $\tilde{\g}_2$ of width $T_1<T_0$. Using the induction hypothesis, we know that $\tilde{\g}_2$ admits a decomposition into bridges of widths $T_1>\cdots>T_j$. The decomposition of $\tilde{\g}$ is created by adding $\tilde{\g}_1$ before the decomposition of $\tilde{\g}_2$. 

If the walk is a reverse half-plane self-avoiding walk, meaning that the end has extremal real part, we set the decomposition to be the decomposition of the reverse walk in the reverse order. If $\g$ is a self-avoiding walk in the plane, one can cut the trajectory into two pieces $\g_1$ and $\g_2$: the vertices of $\g$ up to the first vertex of maximal real part, and the remaining vertices. The decomposition of $\g$ is given by the decomposition of $\g_1$ (with widths $T_{-i}<\cdots <T_{-1}$) plus the decomposition of $\g_2$ (with widths $T_0>\cdots >T_j$). 

Once the starting mid-edge and the first vertex are given, it is easy to check that the decomposition uniquely determines the walk by exhibiting the reverse procedure, see Fig. \ref{fig:decomposition} for the case of half-plane walks.\end{proof}

\begin{remark}\label{remark}
The proof provides bounds for the number of bridges from $a$ to the right side of the strip of width $T$, namely, 
$$\frac{c}{T}\leq B_T^{x_c}\leq 1.$$
In paragraphs 3.3.3 and 3.4.3 of \cite{LawlerSchrammWerner}, precise behaviors are conjectured for the number of self-avoiding walks between two points on the boundary of a domain, which yields the following (conjectured) estimate:
$$\sum_{\gamma\subset S_T:0\rightarrow T+{\rm i}yT} x_c^{\ell(\gamma)} \approx T^{-5/4} H(0,1+{\rm i}y)^{5/4}$$
where $H$ is the boundary derivative of the Poisson kernel. Integrating with respect to $y$, we obtain that $B_T^{x_c}$ should decay as $T^{-1/4}$ when $T$ goes to infinity. Similar estimates are conjectured for walks in $S_T$ from $0$ to ${\rm i}yT$. 
\end{remark}

\section{Conjectures}

In \cite{Nienhuis,Nienhuis-jsp}, Nienhuis proposed a more precise 
asymptotical behavior for the number of self-avoiding walks:
\begin{equation}\label{number}c_n~\sim~ A~n^{\gamma-1}~\sqrt{2+\sqrt 2}^{~n},\end{equation}
with $\gamma=43/32$. 
Here the symbol $\sim$ means that the ratio of two sides 
is of the order $n^{o(1)}$, or perhaps even tends to a constant.
Moreover, Nienhuis gave arguments in support of Flory's prediction
that the
mean-square displacement $\langle |\gamma(n)|^2\rangle$ satisfies
\begin{equation}\label{mean-square}\langle |\gamma(n)|^2\rangle~=~\frac1{c_n}~\sum_{\gamma~n-\text{step SAW}}~|\gamma(n)|^2~=~n^{2\nu+o(1)}~,\end{equation} 
with $\nu=3/4$. 
Despite the precision of the predictions \eqref{number} and \eqref{mean-square}, the best rigorously known bounds are very far apart and almost 50 years old (see \cite{MadrasSlade} for an exposition). The derivation of these exponents seems to be 
one of the most challenging problems in probability. 

It was shown by G.~Lawler, O.~Schramm and W.~Werner in \cite{LawlerSchrammWerner} that $\gamma$ and $\nu$ could be computed 
if the  self-avoiding walk would posses a conformally invariant scaling limit. 
More precisely, let $\Omega\neq \mathbb C$ be a simply connected domain in the complex plane $\mathbb C$ with two points $a$ and $b$ on the boundary. For $\delta>0$, we consider the discrete approximation given by the largest finite domain $\Omega_\delta$ of $\delta\mathbb H$ included in $\Omega$, and $a_\delta$ and $b_\delta$ to be the vertices of $\Omega_\delta$ closest to $a$ and $b$ respectively. A probability measure $\mathbb P_{x,\delta}$ is defined on the set of self-avoiding trajectories $\gamma$ between $a_\delta$ and $b_\delta$ that remain in $\Omega_\delta$ by assigning to $\gamma$ a weight proportional to $x^{\ell(\gamma)}$. We obtain a random 
curve denoted $\gamma_\delta$. Conjectured conformal invariance of self-avoiding walks 
can be stated as follows, see \cite{LawlerSchrammWerner}:

\begin{conjecture}\label{SLE}
Let $\Omega$ be a simply connected domain (not equal to $\mathbb C$) with two distinct points $a$, $b$ on its boundary. For $x=x_c$, the law of $\gamma_\delta$ in $(\Omega_\delta,a_\delta,b_\delta)$ converges when $\delta\rightarrow 0$ to the (chordal) Schramm-Loewner Evolution with parameter $\kappa=8/3$ in $\Omega$ from $a$ to $b$.
\end{conjecture}

As discussed in \cite{lsw-ust,Smirnov}, to prove convergence of a random curve to SLE
it is sufficient to find a discrete observable
with a conformally covariant scaling limit.

Thus it would suffice to show that a normalized version of $F_\delta$
has a conformally invariant scaling limit,
which can be achieved by showing that it is holomorphic
and has prescribed boundary values.

As discussed in \cite{smirnov-icm2010}, the winding of an interface leading to
a boundary edge $z$ is uniquely determined, and coincides with the winding of the boundary
itself.
Thus one can say that $F_\delta$ satisfies a discrete version of the following
\emph{Riemann boundary value problem}
(a homogeneous version of the Riemann-Hilbert-Privalov BVP):
\begin{equation}
{\mathrm{Im}}\,\left({F(z)\cdot\left({\mathrm{tangent~to~}\partial\Omega}\right)^{5/8}}\right)~=~0~,~~~z\in\partial\Omega~,
\label{eq:rbvp}\end{equation}
with a singularity at $a$.
Note that the problem above has conformally covariant solutions
(as  $(dz)^{5/8}$-forms),
and so is well defined even in domains with fractal boundaries.

As noted in Remark~\ref{rem:CR}, relation \eqref{relation around vertex}
amounts to saying that discrete contour integrals of $F_\delta$ vanish.
So any (subsequential) scaling limit of $F_\delta$ would have to be holomorphic.
Unfortunately, relation \eqref{relation around vertex} alone, is unsufficient to
deduce the existence of such a limit, unlike in the Ising case \cite{chelkak-smirnov-iso}.
The reason is that for a domain with $E$ edges, \eqref{relation around vertex} imposes
$\approx \frac23 E$ relations (one per vertice) for $E$ values of $F_\delta$,
making it impossible to reconstruct $F_\delta$ from its boundary values.
So $F_\delta$ is not exactly holomorphic, it can be rather thought of as a divergence-free
vector field, which seems to have non-trivial curl.
However, we expect that in the limit the curl vanishes, which is equivalent to
$F_\delta(z)$ having the same limit regardless of the orientation of the edge $z$.

The Riemann BVP \eqref{eq:rbvp} is easily solved, and
we arrive at the following conjecture:

\begin{conjecture}Let $\Omega$ be a simply connected domain (not equal to $\mathbb C$), let 
$z\in \Omega$, and let $a$, $b$ be two distinct points on the boundary of $\Omega$. We assume that the 
boundary of $\Omega$ is smooth near $b$. For $\delta>0$, let $F_\delta$ be the holomorphic observable in the domain $(\Omega_\delta,a_\delta)$ approximating $(\Omega,a)$, and let $z_\delta$ be the closest point in $\Omega_\delta$ to $z$. Then 

\begin{equation} \lim_{\delta\rightarrow 0}~\frac{F_\delta(z_\delta)}{F_\delta(b_\delta)} ~=~\left(\frac{\phi'(z)}{\phi'(b)}\right)^{5/8}\label{convergence result}\end{equation}
where $\Phi$ is a conformal map from $\Omega$ to the upper half-plane mapping $a$ to 
$\infty$ and $b$ 
to 0. 
\end{conjecture}

The right-hand side of \eqref{convergence result} is well-defined, since the conformal map $\phi$ is 
unique up to multiplication by a real factor. Proving this conjecture would be a major step toward Conjecture~\ref{SLE} and the derivation of critical exponents.

\paragraph{Acknowledgements.} The authors would like to thank G. Slade for useful comments on the manuscript, and G. Lawler for suggesting Remark \ref{remark}. 
This research was supported by the EU Marie-Curie RTN CODY, 
the ERC AG CONFRA, as well as by the Swiss
{FNS}. 
The second author was partially supported by the Chebyshev Laboratory
(Department of Mathematics and Mechanics, St.-Petersburg State University) under
RF governement grant 11.G34.31.0026

\bibliographystyle{alpha}

\begin{flushright}
\footnotesize\obeylines
  \textsc{D\'epartement de Math\'ematiques}
  \textsc{Universit\'e de Gen\`eve}
  \textsc{Gen\`eve, Switzerland}
  \textsc{E-mail:} \texttt{hugo.duminil@unige.ch ; stanislav.smirnov@unige.ch}
\end{flushright}

\end{document}